%% file: arxiv.tex
\documentclass[10pt,table]{article}

\usepackage[utf8]{inputenc}
\usepackage[american]{babel}
\usepackage{enumerate}
\usepackage{amsfonts,amsmath,amssymb,amsthm,amstext,latexsym,paralist}	
\usepackage{hyperref}
\usepackage{xspace}
\usepackage[ruled,vlined,boxed,commentsnumbered]{algorithm2e}
\usepackage{tikz}
\usepackage{subfig}
\usepackage{array}
\usepackage{todonotes}
\usepackage{chemarrow}
\usepackage{ifthen}
\usepackage{forloop}
\usepackage{intcalc}
\usepackage{stmaryrd}
\usepackage{multirow}
\usepackage{fullpage}

\graphicspath{{fig/}}

\input{macro-fig.tex}

\newtheorem{proposition}{Proposition}

\newcommand{\ema}[1]{\ensuremath{#1}\xspace}
\newcommand{\E}{\mathbb{E}}

\def\PP{\mathbb{P}}

\newcommand{\fail}{\mathcal{F}}

\newcommand{\Xlost}{\ema{T_{lost}}}
\newcommand{\Xrec}{\ema{T_{rec}}}
\newcommand{\Rlost}{\ema{R_{lost}}}
\newcommand{\www}{\ema{w}}
\newcommand{\WWW}{\ema{W}}
\newcommand{\lambdae}{\ema{\lambda_{e}}}
\newcommand{\lambdad}{\ema{\lambda_{d}}}
\newcommand{\mue}{\ema{\mu_{e}}}
\newcommand{\mud}{\ema{\mu_{d}}}
\newcommand{\ccc}{\ema{C}}
\newcommand{\rrr}{\ema{R}}
\newcommand{\ddd}{\ema{D}}
\newcommand{\vvv}{\ema{V}}
\newcommand{\sss}{\ema{\mathbb{S}}}
\newcommand{\lamb}{\mathbb{L}} 
\newcommand{\T}{\ensuremath{T}\xspace} 
\newcommand{\Tmin}{\ensuremath{T_{\min}}\xspace} 
\newcommand{\Waste}{\ema{\textsc{Waste}}}
\newcommand{\Wasteff}{\ema{\textsc{Waste}_{\text{ff}}}}
\newcommand{\Wastefail}{\ema{\textsc{Waste}_{\text{fail}}}}
\newcommand{\tbase}{T_{\text{base}}}
\newcommand{\tff}{T_{\text{ff}}}
\newcommand{\tfin}{T_{\text{final}}}
\newcommand{\muplatform}{\mu}
\newcommand{\W}{\ensuremath{\mathit{Work}}\xspace}
\newcommand{\Pfa}{\ema{\mathbb{P}_{\text{fail}}}}
\newcommand{\Pfb}{\ema{\mathbb{P}_{\text{lat}}}}
\newcommand{\Pfc}{\ema{\mathbb{P}_{\text{irrec}}}}
\newcommand{\Pfd}{\ema{\mathbb{P}_{\text{risk}}}}
\newcommand{\risky}{\ema{\varepsilon}}
\newcommand{\Topt}{\ema{T_{\text{opt}}}}

\title{On the Combination of \\ [-.3cm]
Silent Error Detection and Checkpointing}
 
\author{Guillaume Aupy$^{1,3}$, Anne Benoit$^{1,3}$, Thomas H\'erault$^{2}$,\\
 Yves Robert$^{1,2,3}$, Fr\'ed\'eric Vivien$^{3,1}$ and Dounia Zaidouni$^{3,1}$\\~\\
 $1.$ LIP,  \'Ecole Normale Sup\'erieure de Lyon, CNRS, France\\
 $2.$ University of Tennessee Knoxville, USA\\
 $3.$ INRIA
 }

\begin{document}
\maketitle

\begin{abstract}
In this paper, we revisit traditional checkpointing and rollback recovery strategies, 
with a focus on silent data corruption errors. Contrarily to fail-stop failures, 
such latent errors cannot be detected immediately, and a mechanism to detect them
must be provided. We consider two models: (i) errors are detected after some delays following a
probability distribution (typically, an Exponential distribution); (ii) errors are detected through some verification mechanism. In both cases, we compute
the optimal period in order to minimize the waste, i.e., the fraction of time
where nodes do not perform useful computations. In practice, only a fixed number of checkpoints 
can be kept in memory, and the first model may lead to an irrecoverable failure.  
In this case, we compute the minimum period
required for an acceptable risk. For the second model, there is no risk of irrecoverable failure, owing to the
verification mechanism, but the corresponding overhead is included in the waste.
Finally, both models are instantiated using realistic scenarios and application/architecture parameters.
\end{abstract}

\section{Introduction}
\label{sec.intro}

For several decades, the High Performance Computing (HPC) community has been aiming
at increasing the computational capabilities of parallel and distributed platforms,
in order to fulfill expectations  arising from many fields of research, such as chemistry, biology,
medicine and aerospace. The core problem of delivering more performance through ever
larger systems is reliability, because of the number of parallel components. Even if each
independent component is quite reliable, the Mean Time Between Failures (MTBF) 
is expected to drop drastically when considering an exascale system~\cite{IESP-Exascale}. 
Failures become a normal part of application executions. 

The de-facto general-purpose error recovery technique in high performance computing is checkpoint 
and rollback recovery.  Such protocols employ checkpoints to periodically save the state of a 
parallel application, so that when an error strikes some process, the application can be
restored into one of its former states. There are several families of checkpointing protocols. 
We assume in this work that each checkpoint forms a consistent recovery line, i.e.,  
when an error is detected, we
can rollback to the last checkpoint and resume execution, after a downtime and a recovery time. 

Most studies assume instantaneous error detection, 
and therefore apply to fail-stop failures, such as for instance the crash of a resource. 
In this work, we revisit checkpoint protocols in the context of {\em latent} errors,
also called silent data corruption. In HPC, it has been shown recently that such errors
are not unusual, and must also be accounted for~\cite{Moody:2010:DME:1884643.1884666}. 
The cause may be for instance soft efforts in L1 cache, or double bit flips. The problem is that
the detection of a latent error is not immediate, because the error is identified only when 
the corrupted data is activated. One must then account for the detection interval required
to detect the error in the error recovery protocol. Indeed, if the last checkpoint saved
an already corrupted state, it may not be possible to recover from the error. Hence the
necessity to keep several checkpoints so that one can rollback to the last {\em correct}
state. 

This work is motivated by a recent paper by 
Lu, Zheng and Chien~\cite{LuZhengChien2013}, who introduce a \emph{multiple checkpointing model} 
to compute the optimal checkpointing period with error detection latency. 
More precisely, Lu, Zheng and Chien~\cite{LuZhengChien2013} deal with the following problem: 
given errors whose inter arrival times $X_{e}$ follow an Exponential probability distribution 
of parameter \lambdae, and given error detection times $X_{d}$ that follow an Exponential probability 
distribution of parameter \lambdad, what is the optimal
checkpointing period $\Topt$ in order to minimize the total execution time? 
The problem is illustrated on Figure~\ref{fig.gopi1}: the error is detected after a (random) time $X_{d}$, 
and one has to rollback up to the last checkpoint that precedes the occurrence of the error.  
Let $k$ be the number of
checkpoints that can be simultaneously kept in memory.
Lu, Zheng and Chien~\cite{LuZhengChien2013} derive a formula for 
the optimal
checkpointing period $\Topt$ in the (simplified) case where $k$ is unbounded  ($k = \infty$), and they propose some numerical simulations to explore the case where $k$ is a fixed constant.

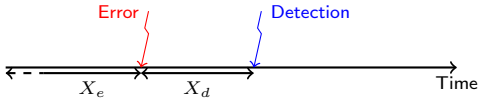
\begin{figure}[t]
\begin{center}
\input{fig/fail-detect-latency.tex}
\end{center}
\vspace{-.3cm}
\caption{Error and detection latency.}
\label{fig.gopi1}
\end{figure}

The first major contribution of this paper is to correct the formula of~\cite{LuZhengChien2013}
when $k$ is unbounded, and to provide an analytical approach when $k$ is a fixed constant. 
The latter approach is a first-order approximation but applies to any probability distribution of errors.

While it is very natural and interesting to consider the latency of error detection, the model 
of~\cite{LuZhengChien2013} suffers from an important limitation: it is not clear how one can 
determine when the error has indeed occurred, and hence to identify the last valid checkpoint, 
unless some verification system is enforced. Another major contribution of this paper is to introduce 
a model coupling verification and checkpointing, and to analytically determine the best balance 
between checkpoints and verifications so as to optimize platform throughput.

The rest of the paper is organized as follows. First we revisit the multiple checkpointing model
of~\cite{LuZhengChien2013} in Section~\ref{sec.chien}; we tackle both the case where all 
checkpoints are kept, and the case with at most $k$ checkpoints. In Section~\ref{sec.ourmodel},
we define and analyze a model coupling checkpoints and verifications. 
Then, we evaluate the various models in Section~\ref{sec.evaluation}, by 
instantiating the models with realistic parameters derived from future exascale platforms. 
Related work is discussed in Section~\ref{sec.related}. 
Finally, we conclude and discuss future research directions in Section~\ref{sec.conclusion}. 

\section{Revisiting the multiple checkpointing model}
\label{sec.chien}

In this section, we revisit the approach of~\cite{LuZhengChien2013}.
We show that their analysis with unbounded memory is incorrect and provide the exact solution 
(Section~\ref{sec.chien-infty}). We also
extend their approach to deal with the case where a given (constant) number of 
checkpoints can be simultaneously kept in memory (Section~\ref{sec.chien-k}).

\subsection{Unlimited checkpoint storage}
\label{sec.chien-infty}

 Let $\ccc$ be the time needed for a checkpoint,  $\rrr$ the time for
recovery, and $\ddd$ the downtime. Although $\rrr$ and $\ccc$ are
a function of the size of the memory footprint of the process, $\ddd$~is a constant 
that represents the unavoidable costs to rejuvenate a
process after an error (e.g., stopping the failed process and restoring
a new one that will load the checkpoint image). We assume that errors
can take place during checkpoint and recovery but not during downtime (otherwise, the downtime
could be considered part of the recovery). 

Let $\mue = \frac{1}{\lambdae}$ be the mean time between errors.
With  no error detection latency and no downtime, well-known formulas for the optimal period 
(useful work plus checkpointing time that minimizes the execution time)
are $\Topt \approx \sqrt{2 \ccc \mue} + \ccc$ (as given by Young~\cite{young74})
and $\Topt \approx \sqrt{2 \ccc (\mue+\rrr)}+ \ccc$ (as given by Daly~\cite{daly04}).
These formulas are first-order approximations and are valid only if $\ccc, \rrr \ll \mue$ 
(in which case they collapse). 

With error detection latency, things are more complicated, even with the assumption
 that one can track the source of the error (and hence identify the last valid checkpoint). 
 Indeed, the amount of rollback will depend upon the sum $X_{e} + X_{d}$.
For Exponential distributions of  $X_{e}$ and $X_{d}$, Lu, Zheng and Chien~\cite{LuZhengChien2013} derive that 
   $\Topt \approx \sqrt{2 \ccc (\mue + \mud) }+ \ccc$, 
  where $\mud = \frac{1}{\lambdad}$ is the mean of error detection times. 
However, although this result may seem intuitive, it is wrong, and we prove that the correct answer
is $\Topt \approx \sqrt{2 \ccc \mue} + \ccc$, even when accounting for the downtime: 
this first-order approximation is the same as Young's formula. We give an intuitive explanation after the proofs
provided in Section~\ref{sec.sub.exp}. Then in Section~\ref{sec.sub.arbi}, we extend this result to
arbitrary laws, but under the additional constraint that $\mud+\ddd+\rrr \ll \mue$.

\subsubsection{Exponential distributions}
\label{sec.sub.exp}

In this section, we assume that $X_{e}$ and $X_{d}$ follow Exponential distributions of mean \mue and \mud respectively.

\begin{proposition}
\label{th.work}
The expected time needed to successfully execute a work of size \www followed by its checkpoint is
$$\E(\T(\www))= e^{\lambdae \rrr} \left(\ddd + \mue + \mud \right) (e^{\lambdae(\www + \ccc)} -1 ).$$
\end{proposition}

\begin{proof}
Let $\T(\www)$ be the time needed for successfully executing a work of duration $\www$.
There are two cases: (i) if there is no error during execution and checkpointing, then the time needed is
exactly $\www+ \ccc$; (ii) if there is an error before successfully completing  the work and its checkpoint, then some additional delays 
are incurred. These delays come from three sources: the time spent
computing by the processors before the error occurs,  the time spent before the error is detected, and the time spent for downtime and recovery. 
Regardless, once a successful recovery has been completed,
 there still remain $\www$ units of work to execute. 
Thus, we can write the following recursion: \\
\begin{equation}
  \E(\T(\www)) = e^{-  \lambdae (\www + \ccc)} (\www+ \ccc)
  + (1-e^{- \lambdae (\www + \ccc)} ) \left( \E(\Xlost)+ \E(X_{d})+\E(\Xrec) +\E(\T(\www)) \right) .
  \label{eq.dq0}
\end{equation}

Here, $\Xlost$ denotes the amount of time spent by the processors before the first error, knowing that this
error occurs within the next $\www + \ccc $ units of time. In other terms, it is the
time that is wasted because computation and checkpoint were not both
completed before the error occurred.
The random variable $X_{d}$ represents the time needed for error detection, and its expectation is 
$\E(X_{d})=\mud=\frac{1}{\lambdad}$.
The last variable $\Xrec$ represents the amount of time needed by the
system to perform a recovery. 
Equation~\eqref{eq.dq0} simplifies to:
\begin{equation}
\small \E(\T(\www))= \www+ \ccc+ (e^{\lambdae (\www + \ccc)} -1 ) ( \E(\Xlost)+ \mud+ \E(\Xrec) ).
\label{eq.dq}
\end{equation}
We have 
{\small 
$$\E(\Xlost) =   \int_0^{\infty} x \PP(X=x|X < \www+\ccc) dx \quad \quad\quad \quad$$
$$\quad \quad= \frac{1}{\PP(X < \www+\ccc)}\int_0^{\www+\ccc} x \lambdae e^{-\lambdae x} dx,$$}
 and $\PP(X < \www+\ccc) =  1-e^{-\lambdae (\www+\ccc)}$.  \\ 
Integrating by parts, we derive that 
\begin{equation}
\E(\Xlost) =  \frac{1}{\lambdae} - \frac{\www+\ccc}{e^{\lambdae  (\www + \ccc)} - 1} .
\label{eq.dq2}
\end{equation}
Next, to compute $\E(\Xrec)$, we have a recursive equation quite similar to Equation~\eqref{eq.dq0}
(remember that we assumed that no error can take place during the downtime):
{\small $$\E(\Xrec)  =  e^{-\lambdae \rrr} (\ddd+\rrr) \quad \quad\quad \quad\quad \quad\quad \quad \quad\quad \quad\quad \quad$$
$$+ (1-e^{-\lambdae \rrr})  (\E(\Rlost) + \E(X_{d}) + \ddd+\E(\Xrec)).$$}
Here, $\E(\Rlost)$ is the expected amount of time lost to executing
the recovery before an error happens, knowing that this
error occurs within the next $\rrr$ units of time. Replacing $\www+\ccc$ by $\rrr$ in Equation~\eqref{eq.dq2},
we obtain
$$\E(\Rlost) =  \frac{1}{\lambdae} - \frac{\rrr}{e^{\lambdae  \rrr} - 1}.$$
The expression for  
$\E(\Xrec)$ simplifies to
\begin{equation}
\E(\Xrec) = \ddd e^{ \lambdae \rrr}  +  (e^{\lambdae \rrr} - 1)(\mue + \mud).
\label{eq.dq3}
\end{equation}
Plugging the values of $\E(\Xlost)$ and $\E(\Xrec)$  into Equation~\eqref{eq.dq} leads to the desired value.
\end{proof}

\begin{proposition}
\label{th.work2}
The optimal strategy to execute a work of size \WWW is to divide it into $n$ equal-size
chunks, each followed by a checkpoint, where $n$ is equal either to  $\max(1,\lfloor n^{*} \rfloor)$ or
to $\lceil n^{*} \rceil$. The value of $n^{*}$ is uniquely derived from $y= \frac{\lambdae \WWW}{n^{*}}-1$, where $\lamb(y) = -e^{-\lambdae \ccc -1}$ ($\lamb$, the Lambert function, 
defined as $\lamb(x)e^{\lamb(x)}=x$).
The optimal strategy does not depend on the value of~\mud.
\end{proposition}

\begin{proof}
Using $n$ chunks of size $\www_{i}$ (with $\sum_{i=1}^{n} \www_{i} = \WWW$), by linearity of the expectation, we have $\E(\T(\WWW))= K \sum_{i=1}^{n}(e^{\lambdae(\www_{i} + \ccc)} -1 )$
where $K = e^{\lambdae \rrr} \left(\ddd + \mue + \mud \right) $ is a constant. By convexity, the sum is minimum when all the
$\www_{i}$s are equal (to $\frac{\WWW}{n}$). Now, $\E(\T(\WWW))$ is a convex function of $n$, hence it admits a unique minimum $n^{*}$ such that the derivative is zero:
\begin{equation}
e^{\lambdae(\frac{\WWW}{n^{*}} + \ccc)} (1 -\frac{\lambdae \WWW}{n^{*}}) = 1.
\label{eq.dq4}
\end{equation}

Let $y=\frac{\lambdae \WWW}{n^{*}}-1$, we have $y e^{y} = -e^{-\lambdae \ccc -1}$, hence
$\lamb(y) = -e^{-\lambdae \ccc -1}$.
Then, since we need an integer number of chunks, the optimal strategy  is
to split $\WWW$ into $\max(1,\lfloor n^{*} \rfloor)$ or
$\lceil n^{*} \rceil$ same-size chunks, whichever leads to the smaller value.
As stated, the value of $y$, hence of $n^{*}$, is independent of $\mud$.
\end{proof}

\begin{proposition}
\label{th.work3}
A first-order approximation for the optimal checkpointing period (that minimizes total execution time) is 
$\Topt \approx \sqrt{2 \ccc \mue}+\ccc$.
This value  is identical to Young's formula, and does not depend on the value of \mud.
\end{proposition}

\begin{proof}
We use  Proposition~\ref{th.work2} and Taylor expansions 
when $z = y+1 = \frac{\lambdae \WWW}{n^{*}}$ is small: from $y e^{y} = -e^{-\lambdae \ccc -1}$,
we derive $(z-1) e^{z} = -e^{-\lambdae \ccc}$. We have  $(z-1) e^{z}\approx \frac{z^{2}}{2}-1$,
and $ -e^{-\lambdae \ccc} \approx -1 + \lambdae \ccc$, hence $z^{2} \approx 2 \lambdae \ccc$. The period
is 
$$\Topt = \frac{\WWW}{n^{*}} + \ccc = \frac{z}{\lambdae} +\ccc
 \approx  \sqrt{2 \ccc \mue}+\ccc.$$
\end{proof}

An intuitive explanation of the result is the following: error
detection latency is paid for every error,
and can be viewed as an additional downtime, which has no impact on the optimal period.

\subsubsection{Arbitrary distributions}
\label{sec.sub.arbi}

Here we extend the previous result to arbitrary distribution laws for $X_{e}$ and $X_{d}$ (of mean \mue and \mud respectively):

\begin{proposition}
\label{th.work4}
When $\ccc \ll \mue$ and $\mud + D + R \ll \mue$, a first-order approximation 
for the optimal checkpointing period is 
$\Topt \approx \sqrt{2 \ccc \mue}+\ccc$.
\end{proposition}

\begin{proof}
Let $\tbase$ be the base time of the application without any overhead due
to resilience techniques.  First, assume a fault-free execution of the
application: every period of length $\T$, only $\W = \T - \ccc$ units of work are executed, 
hence the time $\tff$ for a
fault-free execution is $\tff = \frac{\T}{\W} \tbase$.
Now, let $\tfin$ denote the expectation of the execution time with errors taken into account.  
In average, errors occur every $\mue$
time-units, and for each of them we lose $\fail$ time-units, so there
are $\frac{\tfin}{\mue}$ errors during the execution. Hence we
derive that 
\begin{equation}
\tfin = \tff  + \frac{\tfin}{\mue} \fail , 
\label{eq.tfinal}
\end{equation}
which we rewrite as
\begin{align}
\small
\big(  1\! -\! \Waste \big)& \tfin  = \tbase,\nonumber \\  
 & \small \text{with }
\Waste = 1 - \big(  1 -  \frac{\fail}{\mue} \big)   \big(  1 -  \frac{\ccc}{\T} \big)  .
\label{eq.waste}
\end{align}
The waste is the fraction of time where nodes do not perform useful
computations. Minimizing execution time is equivalent to minimizing the waste. 
In Equation~\eqref{eq.waste}, we identify the two
sources of overhead: (i) the term $\Wasteff = \frac{\ccc}{\T}$, which is the waste due to checkpointing in a fault-free
execution, by construction of the algorithm; and (ii) the term
$\Wastefail = \frac{\fail}{\mue}$, which is the waste due to
errors striking during execution.  With these notations, we have
\begin{equation}
\small
\Waste = \Wastefail  +\Wasteff  - \Wastefail  \Wasteff .
\label{eq.wasteprod}
\end{equation}
There remains to determine the (expected) value of $\fail$. 
Assuming at most one error per period, we lose 
$\fail = \frac{\T}{2}+\mud+\ddd+\rrr$ per error:
$\frac{T}{2}$ for the average work lost before the error occurs, \mud for detecting the error, 
and $\ddd+\rrr$ for downtime and recovery. Note that the assumption is valid only if
$\mud+\ddd+\rrr \ll \mue$ and $\T \ll \mue$.
Plugging back this value into Equation~\eqref{eq.wasteprod},
we obtain 
\begin{equation}
\small
\Waste(\T) = \frac{\T}{2 \mue} + \frac{\ccc(1-\frac{\ddd+\rrr+\mud}{\mue})}{\T}+\frac{\ddd+\rrr+\mud-\frac{\ccc}{2}}{\mue}
\label{eq.wasteT}
\end{equation}
which is minimal for 
\begin{equation}
\small
\Topt = \sqrt{2 \ccc (\mue - \ddd - \rrr - \mud)} \approx \sqrt{2 \ccc \mue} .
\label{eq.tauopt}
\end{equation}
We point out that this approach based on the waste leads to a different approximation
formula for the optimal period,
but $\Topt = \sqrt{2 \ccc (\mue - \ddd - \rrr - \mud)} \approx \sqrt{2 \ccc \mue} \approx \sqrt{2 \ccc \mue} + \ccc$ up to second-order terms, when $\mue$ is large in front of the other parameters, includig~\mud.
For example, this approach does not allow us to handle the case $\mud=\mue$; in such a case, the optimal
period is known only for Exponential distributions, and is independent of \mud, as proven by
Proposition~\ref{th.work2}.
\end{proof}

To summarize, the exact value of the optimal period is only known for Exponential distributions and is provided  by Proposition~\ref{th.work2},  while  Young's formula can be used as a first-order approximation
for any other distributions.  Indeed, the
optimal period is a trade-off
between the overhead due to checkpointing ($\frac{\ccc}{T}$) and the expected time lost per error
($\frac{T}{2 \mue}$ plus some constant). Up to second-order terms, the waste is minimum when both factors are equal, which leads to Young's formula, and which remains valid regardless of error detection latencies.

\subsection{Saving only $k$ checkpoints}
\label{sec.chien-k}

Lu, Zheng and Chien~\cite{LuZhengChien2013} propose a set of simulations to assess 
the overhead induced when keeping only the last $k$ checkpoints (because of storage limitations). 
In the following, we derive an analytical approach to numerically solve the problem. 
The main difficulty is that when error detection latency is too large, it is impossible 
to recover from a valid checkpoint, and one must resume the execution from scratch. 
We consider this scenario as an \emph{irrecoverable failure}, and we aim at guaranteeing 
that the risk of  irrecoverable failure remains under a user-given threshold.

Assume that a job of total size \WWW is partitioned into $n$ chunks. 
What is the risk of irrecoverable failure during the execution of one chunk of size $\frac{\WWW}{n}$ followed by its checkpoint?
 Let $\T = \frac{\WWW}{n} + \ccc$ be the length of the period. Intuitively, the longer the period, 
 the smaller the probability that
an error that has just been detected took place more than $k$ periods ago, thereby leading to 
an irrecoverable failure because the last valid checkpoint is not one of the $k$ most recent ones. 

Formally, there is an irrecoverable failure if: 
(i) there is an error detected during the period (probability \Pfa),  and 
(ii) the sum of $\Xlost$, the time elapsed since the last checkpoint, and of  $X_{d}$, 
the error detection latency, exceeds  $k\T$ (probability \Pfb). 
The value of $\Pfa = \PP(X_{e}\leq \T)$ is easy to compute from the error distribution law. For instance with an Exponential law, $\Pfa = 1 - e^{- \lambdae \T}$. As for \Pfb, we use an upper bound:
$\Pfb = \PP(\Xlost+X_{d}\geq k\T) \leq \PP(\T+X_{d}\geq k\T) = \PP(X_{d}\geq (k-1)\T) $. The latter
value is easy to compute from the error distribution law. For instance with an Exponential law, 
$\Pfb = e^{- \lambdad (k-1) \T}$. Of course, if there is an error and the error detection latency 
does not exceed  $k\T$ (probability (1-\Pfb)), we have to restart execution and face the same risk as before.
Therefore, the probability of irrecoverable failure \Pfc can be recursively evaluated as
$\Pfc = \Pfa (\Pfb +  (1-\Pfb)  \Pfc)$, hence $\Pfc = \frac{\Pfa \Pfb}{1- \Pfa (1-\Pfb)}$. Now that we have computed \Pfc, the probability of irrecoverable  failure for a single chunk,
we can compute the probability of irrecoverable failure for  $n$ chunks as $\Pfd = 1 - (1 - \Pfc)^{n}$.
In full rigor,
these expressions for \Pfc and \Pfd are valid only for Exponential distributions, because of the memoryless property, but
they are a good approximation for arbitrary laws. Given a prescribed risk threshold $\risky$,
solving numerically the equation $\Pfd \leq \risky$ leads to a
lower bound $\Tmin$ on $\T$. Let $\Topt$ be the optimal period given in Theorem~\ref{th.work3}
for an unbounded number of saved checkpoints. The best strategy is then to use the period $\max(\Tmin, \Topt)$ to minimize the waste while enforcing a risk below threshold. 

In case of irrecoverable  failure, we have to resume execution from the very beginning. The number of re-executions
due to consecutive irrecoverable  failures follows a geometric law of parameter $1-\Pfd$, so that the expected number of executions until success is $\frac{1}{1 - \Pfd}$.
We refer to Section~\ref{sec.evaluation.k} for an example of how to instantiate this model to
compute the best period with a fixed number of checkpoints, under a prescribed risk threshold.

\section{Coupling verification and checkpointing}
\label{sec.ourmodel}

In this section, we move to a more realistic model where silent errors are detected only when some verification mechanism (checksum, error correcting code, coherence tests, etc.) is executed. Our approach is agnostic of the nature of this verification mechanism. We aim at solving the following optimization problem:
given the cost of checkpointing \ccc, downtime \ddd, recovery \rrr, and verification \vvv, what is the optimal 
strategy to minimize the expected waste as a function of the mean time between errors \mue? Depending upon
the relative costs of checkpointing and verifying, we may have more checkpoints than verifications, 
or the other way round. In both cases, we target a periodic pattern that repeats over time. 

Consider first the scenario where the cost of a checkpoint is smaller than the cost of 
a verification: then the periodic pattern will include $k$
checkpoints and $1$ verification, where $k$ is some parameter to determine. 
Figure~\ref{fig.gopi2}(a) provides an illustration with $k=5$. We assume that the 
verification is directly followed by the last checkpoint in the pattern, so as to save results 
just after they have been verified (and before they get corrupted). In this scenario,
the objective is to determine the value of $k$ that leads to the minimum platform waste.
This problem is addressed in Section~\ref{sec.kc1v}.

Because our approach is agnostic of the cost of the verification, we also envision
scenarios where the cost of a checkpoint is higher than the cost of 
a verification. In such a framework, the periodic pattern will include $k$ verifications
and $1$ checkpoint, where $k$ is some parameter to determine. 
See Figure~\ref{fig.gopi2}(b) for an illustration with $k=5$. Again, the 
objective is to determine the value of $k$ that leads to the minimum platform waste.
This problem is addressed in Section~\ref{sec.kv1c}.

We point out that combining verification and checkpointing guarantees that no irrecoverable  failure
will kill the application: the last checkpoint of any period pattern is always correct, because
a verification always takes place right before this checkpoint is taken. If that verification reveals an error, we roll back
until reaching a correct verification point, maybe up to the end of the previous pattern, but never further back,
and re-execute the work.
The amount of roll-back and re-execution depends upon the shape of the pattern,
and we show how to compute it in Sections~\ref{sec.kc1v} and~\ref{sec.kv1c} below.

\begin{figure}[htbp]
\subfloat[$5$ checkpoints for $1$ verification]{
\resizebox{1\linewidth}{!}{
\input{fig/pattern5C1V.tex}
}}

\subfloat[$5$ verifications for $1$ checkpoint]{
\resizebox{1\linewidth}{!}{
\input{fig/pattern1C5V.tex}
}}
\caption{Periodic pattern.}
\label{fig.gopi2}
\end{figure}
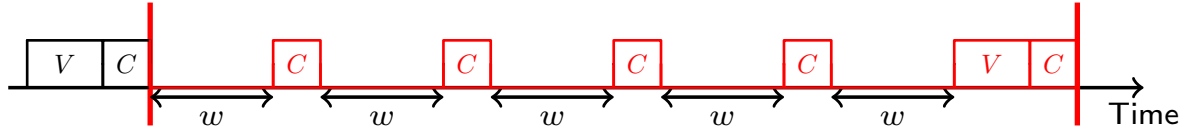
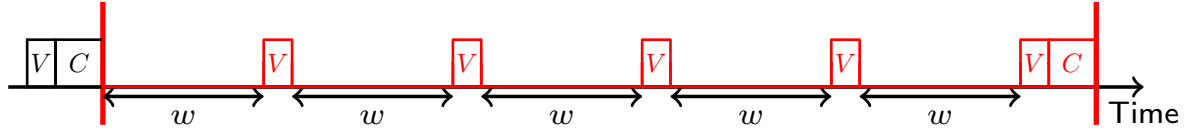

\vspace*{-0.3cm}
\subsection{With $k$ checkpoints and $1$ verification}
\label{sec.kc1v}

We use the same approach as in the proof of Proposition~\ref{th.work4} and compute 
a first-order approximation of the waste (see Equations~\eqref{eq.waste} and~\eqref{eq.wasteprod}).
We compute the two
sources of overhead: (i) \Wasteff, the waste incurred in a fault-free
execution, by construction of the algorithm, and (ii)
\Wastefail, the waste due to errors striking during execution. 

Let $\sss= k\www + k\ccc + \vvv$ be the length of the periodic pattern.
We easily derive that $\Wasteff = \frac{k\ccc + \vvv}{\sss}$. As for \Wastefail, we still have
 $\Wastefail = \frac{\ddd+\E(\Xlost)}{\mue}$. However, in this context, the time lost because of the error depends upon 
the location of this error within the periodic pattern, so we compute averaged values as follows.
Recall (see Figure~\ref{fig.gopi2}(a)) that checkpoint $k$ is the one preceded by a verification. 
Here is the analysis when an error is detected during the verification that takes place in the pattern:
\begin{compactitem}
\item If the error took place in the (last) segment $k$: we recover from checkpoint $k-1$, 
and verify it; we get a correct result because the error took place
later on. Then we re-execute the last piece of work and redo the
verification. The time that has been lost is $\Xlost(k) = \rrr + \vvv
+ \www + \vvv$. (We assume that there is at most one error
per pattern.)
\item If the error took place in segment $i$, $2 \leq i \leq k-1$:  we recover from checkpoint $k-1$, 
verify it, get a wrong result; we iterate, going back up to checkpoint $i-1$, verify it, and get a correct result because the error took place later on. Then we re-execute $k-i+1$ pieces of work and $k-i$ checkpoints, together with the last verification. We get $\Xlost(i) = (k-i+1) (\rrr + \vvv + \www) + (k-i)\ccc + \vvv$.
\item If the error took place in (first) segment $1$: this is almost the same as above, except that the first recovery at the beginning of the pattern need not be verified, because the verification was made just before
the corresponding checkpoint at the end of the previous pattern. We have the same formula with $i=1$ but with one fewer verification:  $\Xlost(1) = k(\rrr + \www) + (k-1)(\ccc + \vvv) + \vvv$.
\end{compactitem}
Therefore, the formula for \Wastefail writes
\begin{equation}
\small
\Wastefail = \frac{\ddd+\frac{1}{k}\sum_{i=1}^{k} \Xlost(i)}{\mue} , 
\label{eq.waste-fail-kc1v}
\end{equation}
and (after some manipulation using a computer algebra system) 
the formula simplifies to 
\begin{equation}
\small
\Wastefail = \frac{1}{2 k \mue}((\rrr\!+\!\vvv)k^{2} \!+\!(2\ddd\!+\!\rrr\!+\!2\vvv\!+\!\sss\!-\!2\ccc)k\!+\!\sss\!-\!3\vvv)
\label{eq.wastefail-kc1v}
\end{equation}
Using $\Wasteff = \frac{k\ccc + \vvv}{\sss}$ and Equation~\eqref{eq.wasteprod}, we compute
the total waste and derive that $\Waste = a \sss + b +  \frac{c}{\sss}$, where $a$, $b$, and $c$ are
some constants. The optimal value of $\sss$ is $\sss_{opt} = \sqrt{\frac{c}{a}}$, provided that this value
is at least $k\ccc + \vvv$. 
We point out that this formula only is a first-order approximation. We have assumed a single error
per pattern. We have also assumed that errors did not occur during
checkpoints following verifications.
Now, once we have found $\Waste(\sss_{opt})$, the value of the waste obtained
for the optimal period $\sss_{opt}$, we can minimize this quantity as a function of $k$, and 
numerically derive the
optimal value $k_{opt}$ that provides the best value (and hence the best platform usage). 

Due to lack of space, computational details are available in~\cite{webrefmaple}, which is
a Maple sheet that we have to instantiate the model. This Maple sheet is publicly available for
users to experiment with their own parameters. We provide two example scenarios to
illustrate the model in Section~\ref{sec.evaluation.kc1v}.

Finally, note that in order to minimize the waste, one could do a binary search in order to find the last
checkpoint before the fault. Then we can upper-bound $\Xlost(i)$ by $(k-i+1) \www+ \log(k)(\rrr+\vvv) + 
(k-i)\ccc + \vvv$, and Equation~\eqref{eq.wastefail-kc1v} becomes $\Wastefail = \frac{1}{2 k \mue}((\rrr+\vvv)2k \log(k) +(2\ddd+\rrr+2\vvv+\sss-2\ccc)k+\sss-3\vvv) $.

\subsection{With $k$ verifications and $1$ checkpoint}
\label{sec.kv1c}

We use  a similar line of reasoning for this scenario and compute 
a first-order approximation of the waste for the case with $k$ verifications and $1$ checkpoint per 
pattern. The length of the periodic pattern is now $\sss= k \www + k \vvv + \ccc$.
As before, for $1 \leq i \leq k$, let segment $i$ denote the period of work before 
verification $i$, and assume (see Figure~\ref{fig.gopi2}(b)) that verification $k$  is preceded by a
checkpoint.  The analysis is somewhat simpler here.

If an error takes place in segment $i$, $1 \leq i \leq k$, we detect the error during verification~$i$,
we recover from the last checkpoint, and redo the first $i$ segments and verifications: therefore
$\Xlost(i) = \rrr + i(\vvv + \www )$.
The formula for \Wastefail is  the same as in Equation~\eqref{eq.waste-fail-kc1v}
and (after some manipulation) we derive
\begin{equation}
\Wastefail = \frac{1}{2 \mue}\left(\ddd+\rrr+\frac{k+1}{2k}\left(\sss-\ccc\right)\right) .
\label{eq.wastefail-kv1c}
\end{equation}
Using $\Wasteff = \frac{k\vvv + \ccc}{\sss}$ and Equation~\eqref{eq.wasteprod}, we proceed just as
in Section~\ref{sec.kc1v} to compute the optimal value $\sss_{opt}$ of the periodic pattern, and then
the optimal value $k_{opt}$ that minimizes the waste. Details are available within the 
Maple sheet~\cite{webrefmaple}.

\section{Evaluation}
\label{sec.evaluation}

This section provides some examples for instantiating the various
models. We aimed at choosing realistic parameters in the context of
future exascale platforms, but we had to restrict to a limited set of
scenarios, which do not intend to cover the whole spectrum of possible
parameters. The Maple sheet~\cite{webrefmaple} is available to explore
other scenarios.

\subsection{Best period with $k$ checkpoints under a given risk threshold}
\label{sec.evaluation.k}

We first evaluate $\Pfd$, the risk of irrecoverable failure, as defined in Section~\ref{sec.chien-k}.
Figures~\ref{fig:maple:risk1} and~\ref{fig:maple:risk2}
present, for different scenarios, the probability $\Pfd$ as a function
of the checkpointing period $\T$ on the left. On the right, the
figures present the corresponding
waste with $k$ checkpoints and 
in the absence of irrecoverable  failures. This waste can be computed  
following the same reasoning as in Equation~\eqref{eq.wasteT}.
For each figure, the left
diagram represents the risk implied by a given period $\T$, showing
the value $\Topt$ of the optimal checkpoint interval (optimal
with respect to waste minimization and in the absence of irrecoverable  failures, 
see Equation~\eqref{eq.tauopt}) as a blue vertical line. The right
diagram on the figure represents the corresponding waste, highlighting
the trade-off between an increased irrecoverable-failure-free waste and a reduced risk. As stated
in Section~\ref{sec.chien-k}, it does not make sense to select a
value for $T$ lower than $\Topt$, since the waste would be
increased, for an increased risk.

\begin{figure}
\begin{center}
\includegraphics[width=.45\linewidth]{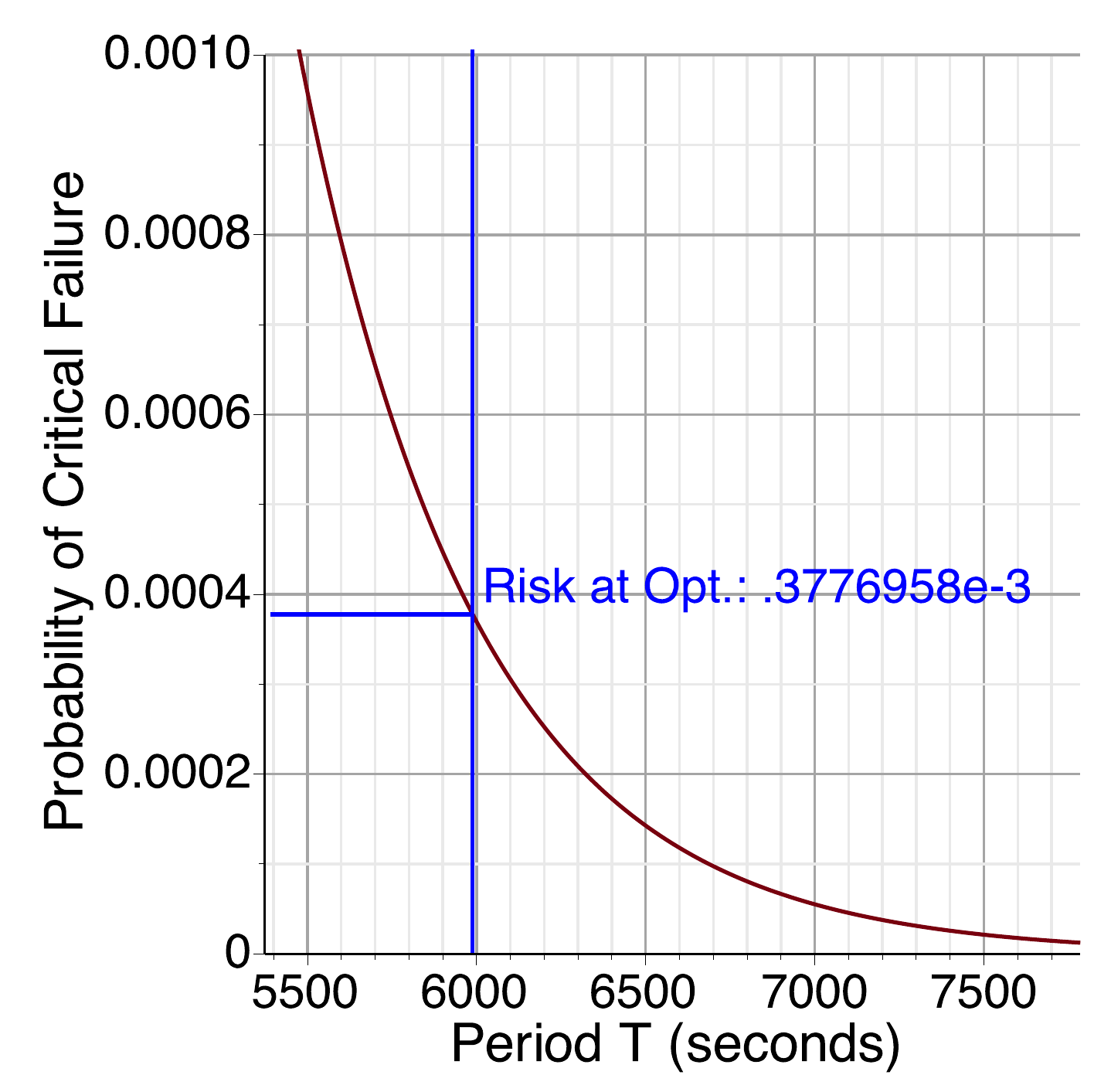}
\includegraphics[width=.45\linewidth]{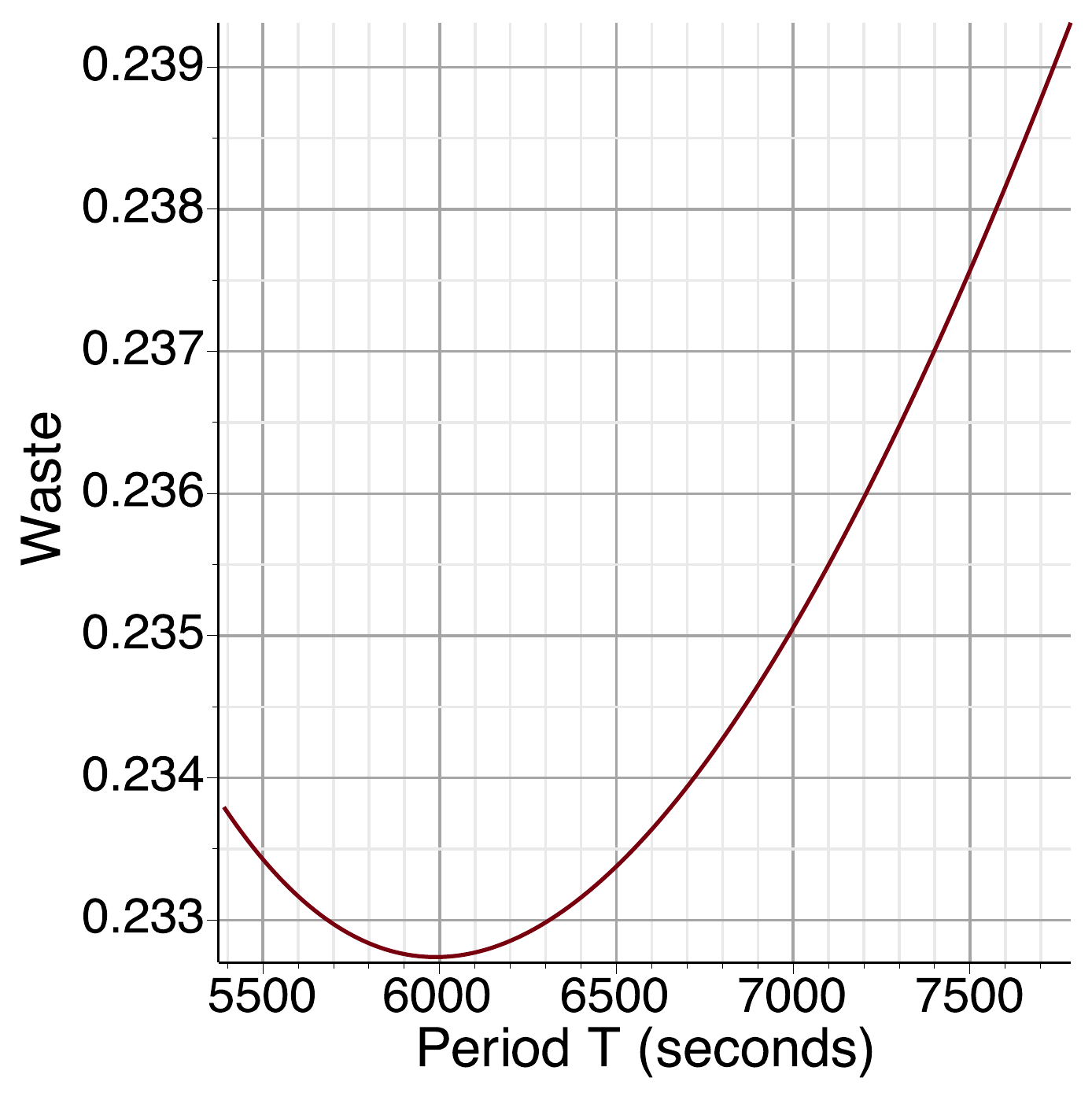}
\caption{Risk of irrecoverable failure as a function of the checkpointing
  period, and corresponding waste. {\footnotesize($k=3$, $\lambdae\!=\!\frac{10^5}{100y}, \lambdad=30\lambdae,
  \www=10d, \ccc=\rrr=600s,$ and $\ddd=0s$.)}}
\label{fig:maple:risk1}
\end{center}
\end{figure}

Figure~\ref{fig:maple:risk1} considers a machine consisting of $10^5$
components, and a component MTBF of 100 years. This component MTBF
corresponds to the optimistic assumption on the reliability of
computers made in the literature~\cite{IESP-toward,IESP-Exascale}.
The platform MTBF $\mue$ is thus $100\times 365\times 24/100,000 \approx 8.76$
hours.  The times to checkpoint and recover (10 min) correspond to
reasonable mean values for systems at this
size~\cite{c178,Ferreira2011}. At this scale, process rejuvenation is
small, and we set the downtime to 0s. For these average values to have a
meaning, we consider a run that is long enough (10 days of work), and
in order to illustrate the trade-off, we take a rather low (but
reasonable) value $k=3$ of intervals, and a mean time error
detection $\mud$ significantly smaller (30 times) than the MTBF $\mue$
itself.

With these parameters, $\Topt$ is around 100 minutes, and the
risk of irrecoverable failure at this checkpoint interval can be evaluated
at $1/2617\approx 38\cdot 10^{-5}$, inducing an irrecoverable-failure-free waste of $23.45\%$. To
reduce the risk to $10^{-4}$, a $T_{\min}$ of $8000$ seconds is sufficient,
increasing the waste by only $0.6\%$. In this case, the benefit of
fixing the period to $\max(\Topt, T_{\min})$ is obvious. Naturally,
keeping a bigger amount of checkpoints (increasing $k$) would also
reduce the risk, at constant waste, if memory can be afforded.

\begin{figure}
\begin{center}
\includegraphics[width=.45\linewidth]{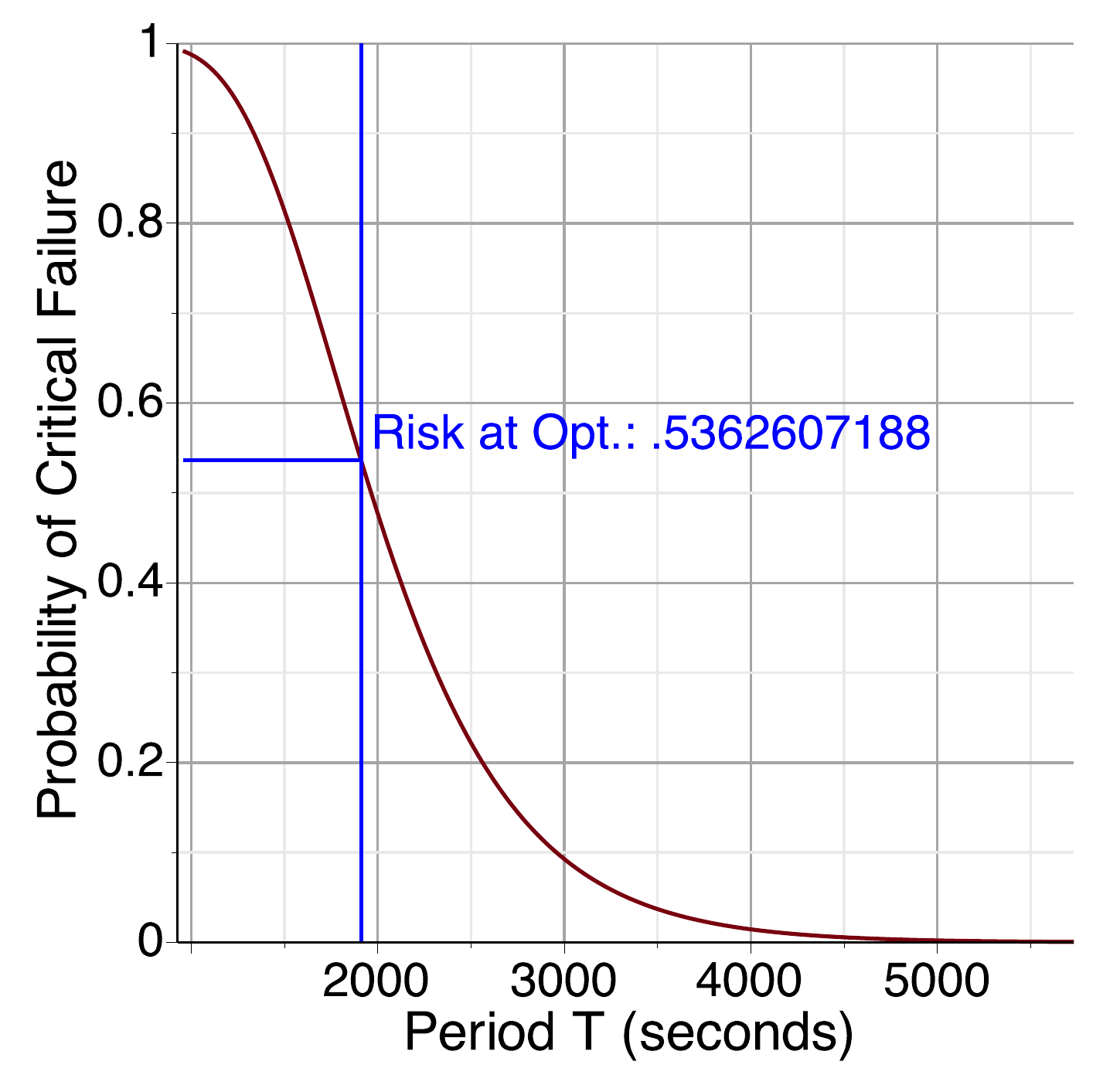}
\includegraphics[width=.45\linewidth]{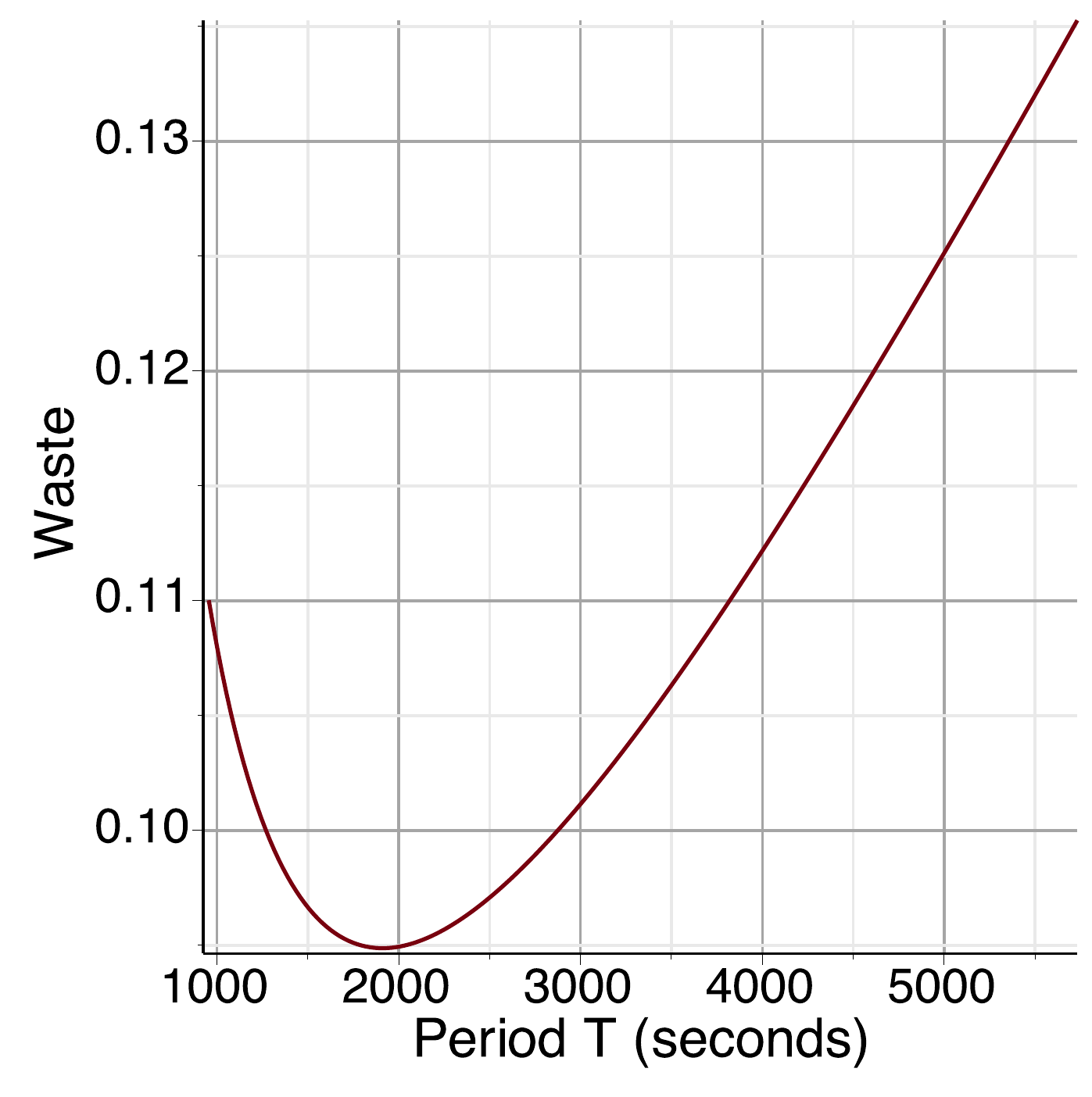}
\caption{Risk of irrecoverable failure as a function of the checkpointing
  period, and corresponding waste. {\footnotesize($k=3$, $\lambdae=\frac{10^5}{100y}, \lambdad=30\lambdae,
  \www=10d, \ccc=\rrr=60s,$ and $\ddd=0s$.)}}
\vspace{-.6cm}
\label{fig:maple:risk2}
\end{center}
\end{figure}

We also consider in Figure~\ref{fig:maple:risk2} a more optimistic scenario where the
checkpointing technology and availability of resources is increased by
a factor 10: the time to checkpoint, recover, and allocate new
computing resources is divided by 10 compared to the previous
scenario. Other parameters are kept similar. One can observe that
$\Topt$ is largely reduced (down to less than 35 minutes between checkpoints), 
as well as the optimal irrecoverable-failure-free waste ($9.55\%$). 
This is unsurprising, and mostly due to the reduction of
failure-free waste implied by the reduction of checkpointing time. But
because the period between checkpoints becomes smaller, while the
latency to detect an error is unchanged ($\mud$ is still 30 times smaller than
 $\mue$), the risk that an error happens at the
interval $i$ but is detected after interval $i+k$ is increased. Thus,
the risk climbs to $1/2$, an unacceptable value. To reduce the risk
to $10^{-4}$ as previously, it becomes necessary to consider a
$T_{min}$ of $6650$ seconds, which implies an irrecoverable-failure-free waste of $15\%$,
significantly higher than the optimal one, which is  below $10\%$, but still
much lower than the $24\%$ when checkpoint and
availability costs are 10 times higher.

\subsection{Periodic pattern with $k$ verifications and $1$ checkpoint}
\label{sec.evaluation.kv1c}

We now focus on the waste induced by the different ways of coupling
periodic verification and checkpointing. We first consider the case of a
periodic pattern with more verifications than checkpoints: every $k$
verifications of the current state of the application, a
checkpoint is taken. The duration of the work interval \sss, between
two verifications in this case, is optimized to minimize the
waste. We consider two scenarios. For each scenario, we represent two diagrams: the left diagram
shows the waste as a function of $k$ for a given verification cost \vvv,
and the right diagram shows the waste as a function of $k$ and \vvv
using a 3D surface representation.

\begin{figure}
\begin{center}
\includegraphics[width=.45\linewidth]{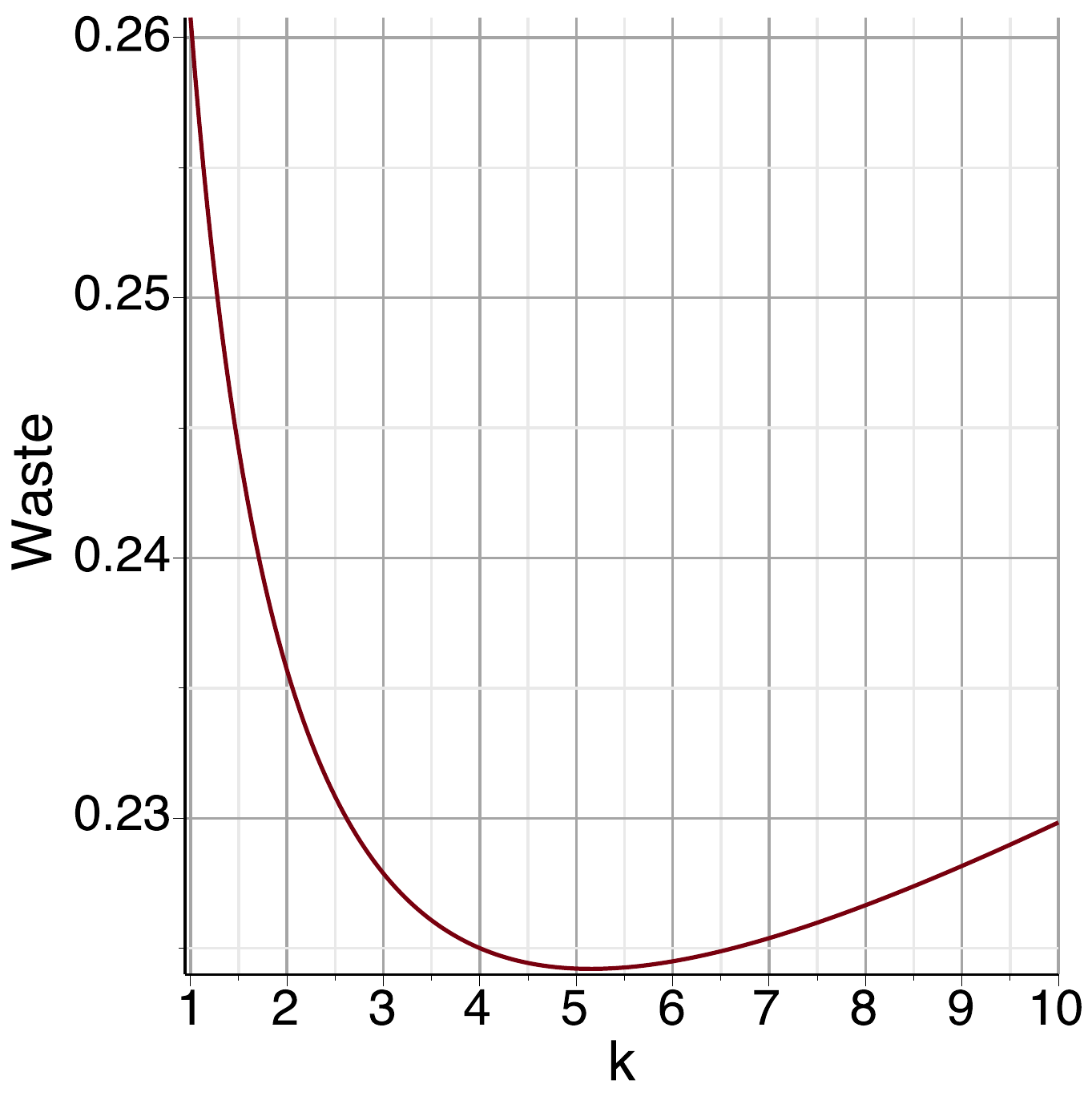}
\includegraphics[width=.45\linewidth]{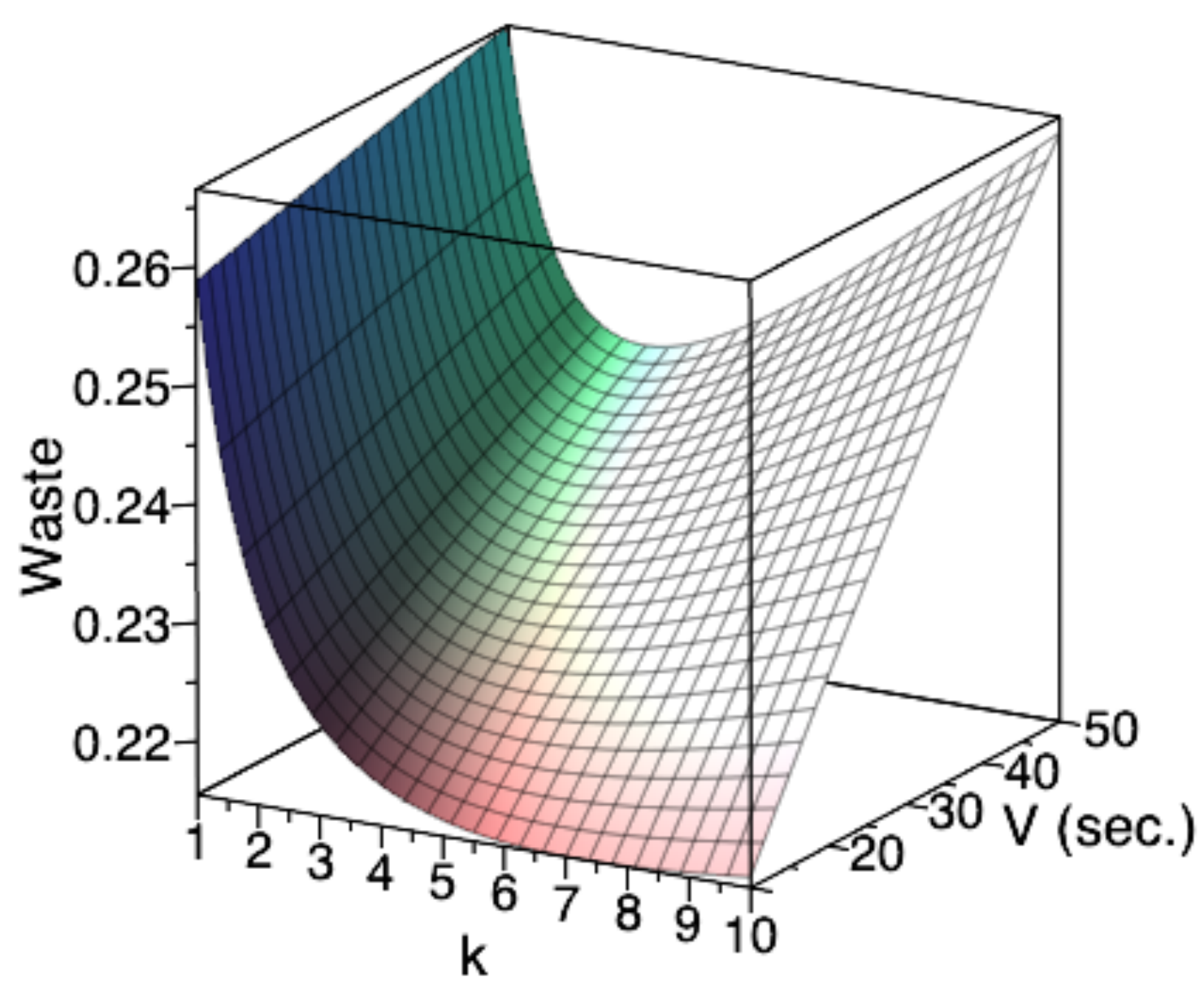}
\caption{Case with $k$ verifications, and one checkpoint per
periodic  pattern. Waste as function of $k$, and potentially of $V$,
using the optimal period.
\footnotesize{($\vvv=20s, \ccc = \rrr = 600s, \ddd = 0s$, and $\muplatform = \frac{10y}{10^5}$.)}
\vspace{-2\baselineskip}
\label{fig:maple:kV1C1}}
\end{center}
\end{figure}

In the first scenario, we consider the same setup as above in
Section~\ref{sec.evaluation.k}.  The waste is computed in its general
form, so we do not need to define the duration of the work. As
represented in Figure~\ref{fig:maple:kV1C1}, for a given verification
cost, the waste can be optimized by making more than one 
verifications. When $k > 1$, there are intermediate verifications that
can enable to detect an error before a periodic pattern (of length
$\sss$) is completed, hence, that can reduce the time lost due to an
error.
However, introducing
too many verifications induces an overhead that eventually dominates the
waste. The 3D surface shows that the waste reduction is significant
when increasing the number of verifications, until the optimal number
is reached. Then, the waste starts to increase again
slowly. Intuitively, the lower the cost for \vvv, the higher the
optimal value for~$k$.

\begin{figure}
\begin{center}
\includegraphics[width=.45\linewidth]{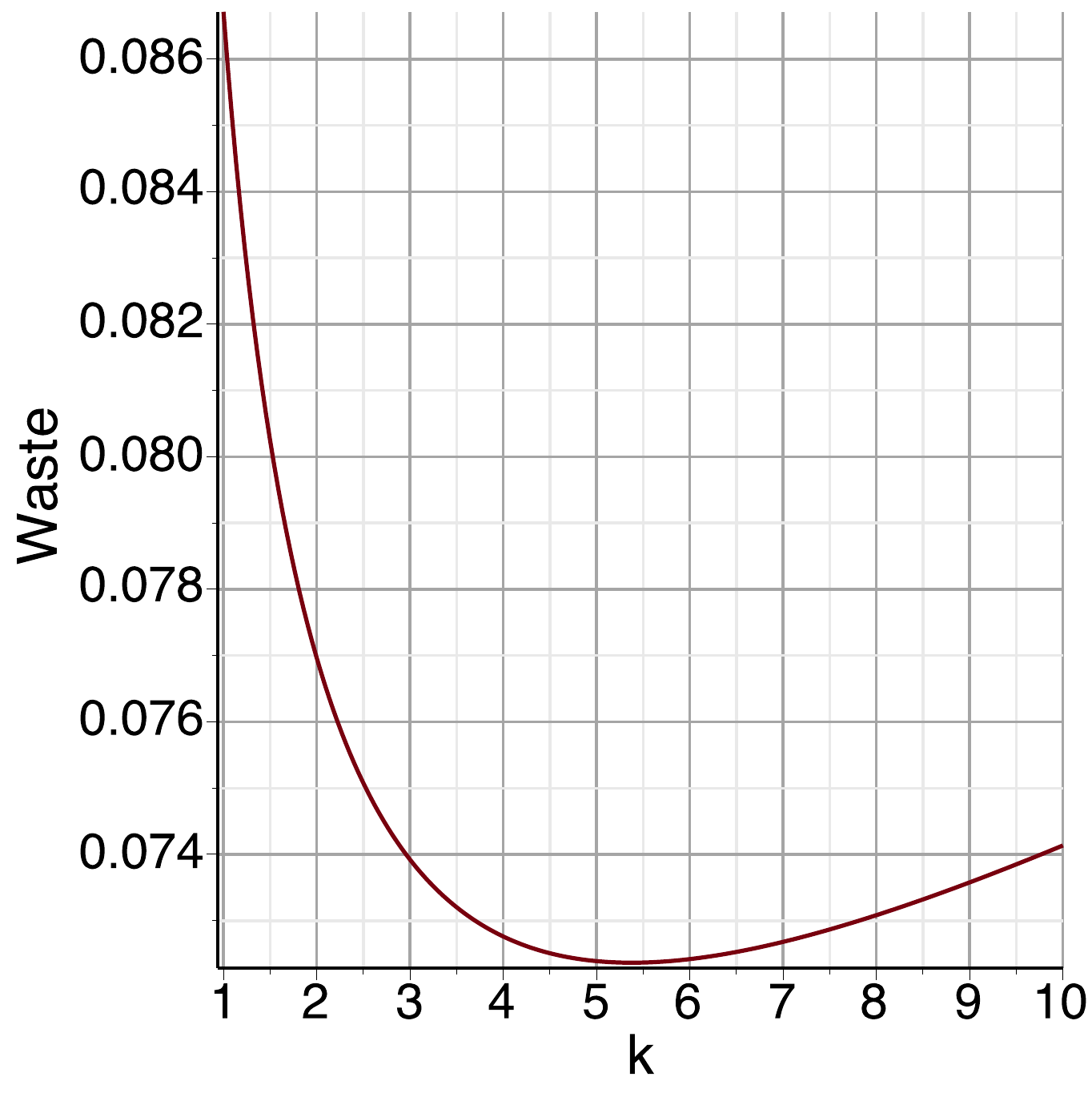}
\includegraphics[width=.45\linewidth]{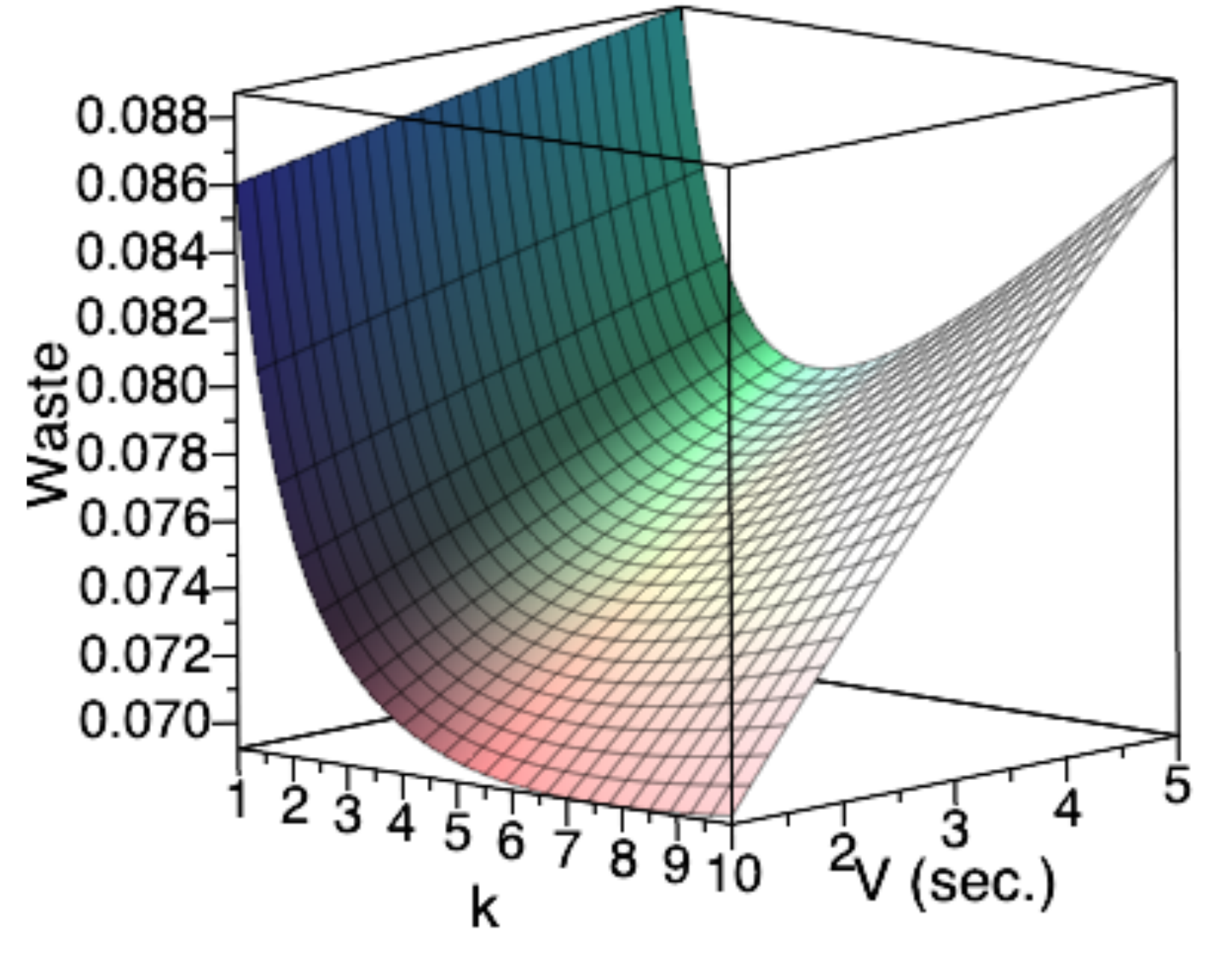}
\caption{Case with $k$ verifications, and one checkpoint per
  periodic  pattern. Waste as function of $k$, and potentially of $V$,
  using the optimal period.
  \footnotesize{($\vvv=2s, \ccc=\rrr=60s,
  \ddd=0s$, and $\muplatform =
  \frac{10y}{10^5}$.)}
\label{fig:maple:kV1C2}\vspace{-1\baselineskip}}
\end{center}
\end{figure}

When considering the second scenario (Figure~\ref{fig:maple:kV1C2}),
with an improved checkpointing and availability setup, the same
conclusions can be reached, with an absolute value of the waste
greatly diminished. Since forced verifications allow to detect the
occurrence of errors at a controllable rate (depending on \sss and
$k$), the risk of non-recoverable errors is nonexistent in this case,
and the waste can be greatly diminished, with very few checkpoints
taken and kept during the execution.

\subsection{Periodic pattern with $k$ checkpoints and $1$ verification}
\label{sec.evaluation.kc1v}

The last set of experiments considers the opposite case of periodic
patterns: checkpoints are taken more often than verifications. Every
$k$ checkpoints, a verification of the data consistency is
done. Intuitively, this could be useful if the cost of verification is
large compared to the cost of checkpointing itself. In that case, when
rolling back after an error is discovered, each checkpoint that was
not validated before is validated at rollback time, potentially
invalidating up to $k-1$ checkpoints.

\begin{figure}
\begin{center}
\includegraphics[width=.45\linewidth]{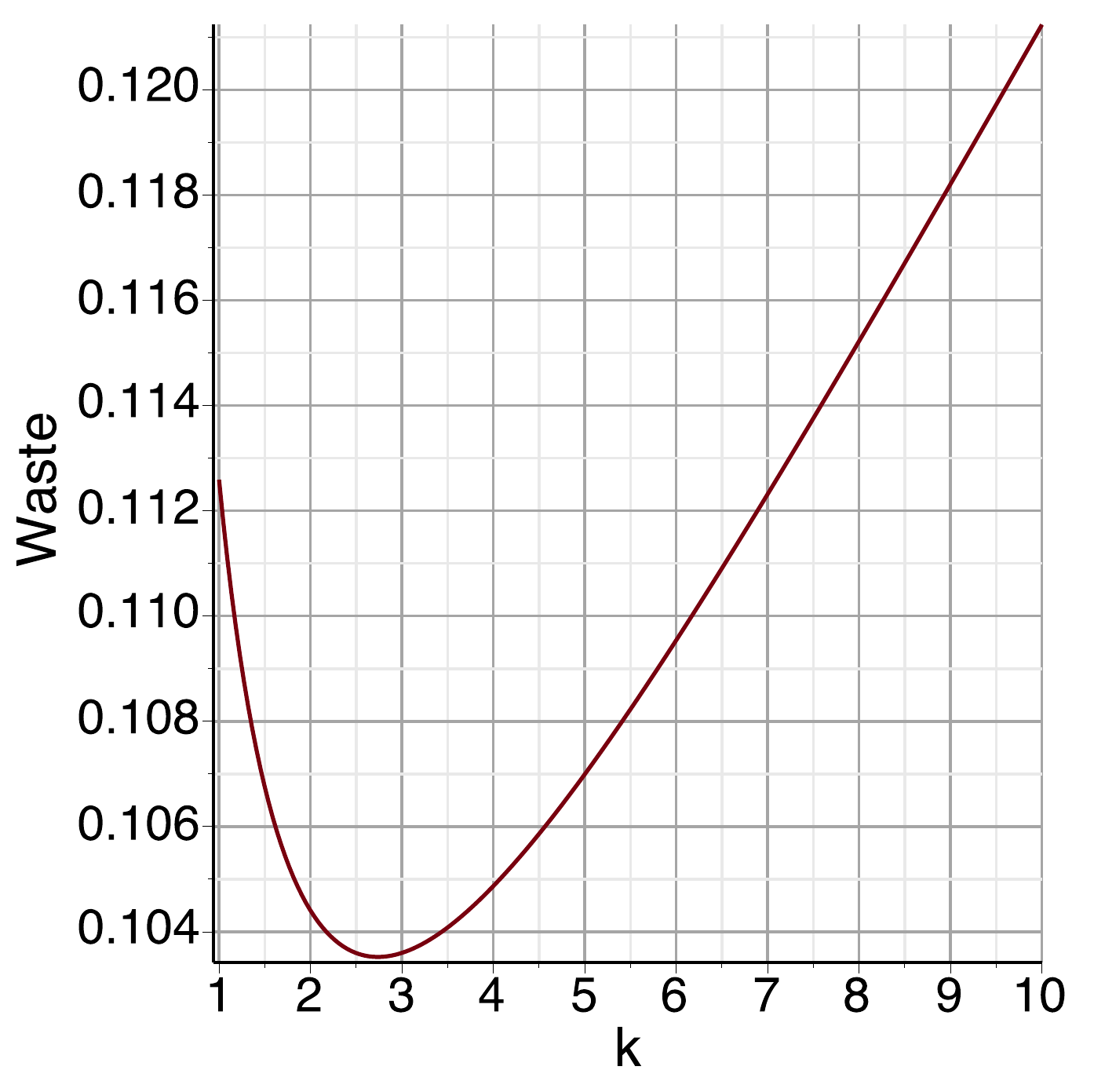}
\includegraphics[width=.45\linewidth]{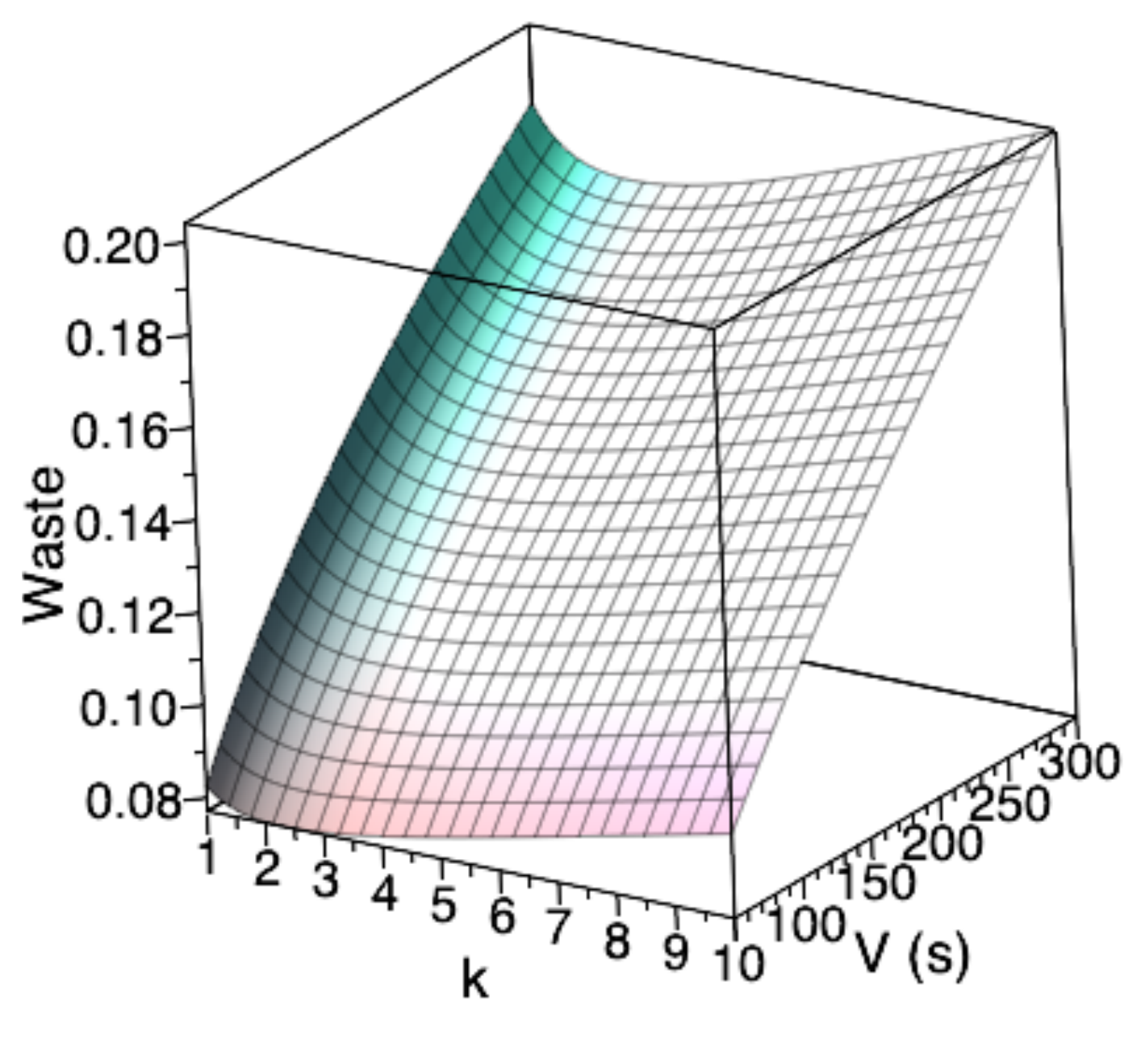}
\caption{Case with $k$ checkpoints, and one verification per periodic
  pattern. Waste as function of $k$, and potentially of $V$, using the
  optimal period.  \footnotesize{($\vvv=100s, \ccc=\rrr=6s, \ddd=0s$,
    and $\muplatform = \frac{10y}{10^5}$.)}
\label{fig:maple:kC1V1}\vspace{-\baselineskip}}
\end{center}
\end{figure}

Because this pattern has potential only when the cost of checkpoint is
much lower than the cost of verification, we considered the case of a
greatly improved checkpoint / availability setup: the checkpoint and
recovery times are only $6$ seconds in
Figure~\ref{fig:maple:kC1V1}. One can observe that in this extreme
case, it can still make sense to consider multiple checkpoints between
two verifications (when $\vvv=100$ seconds, a verification is done
only every 3 checkpoints optimally); however the 3D surface
demonstrates that the waste is still dominated by the cost of the
verification, and little improvement can be achieved by taking the
optimal value for~$k$. The cost of verification must be incurred when rolling
back, and this shows on the overall performance if the verification is
costly.

\begin{figure}
\begin{center}
\includegraphics[width=.45\linewidth]{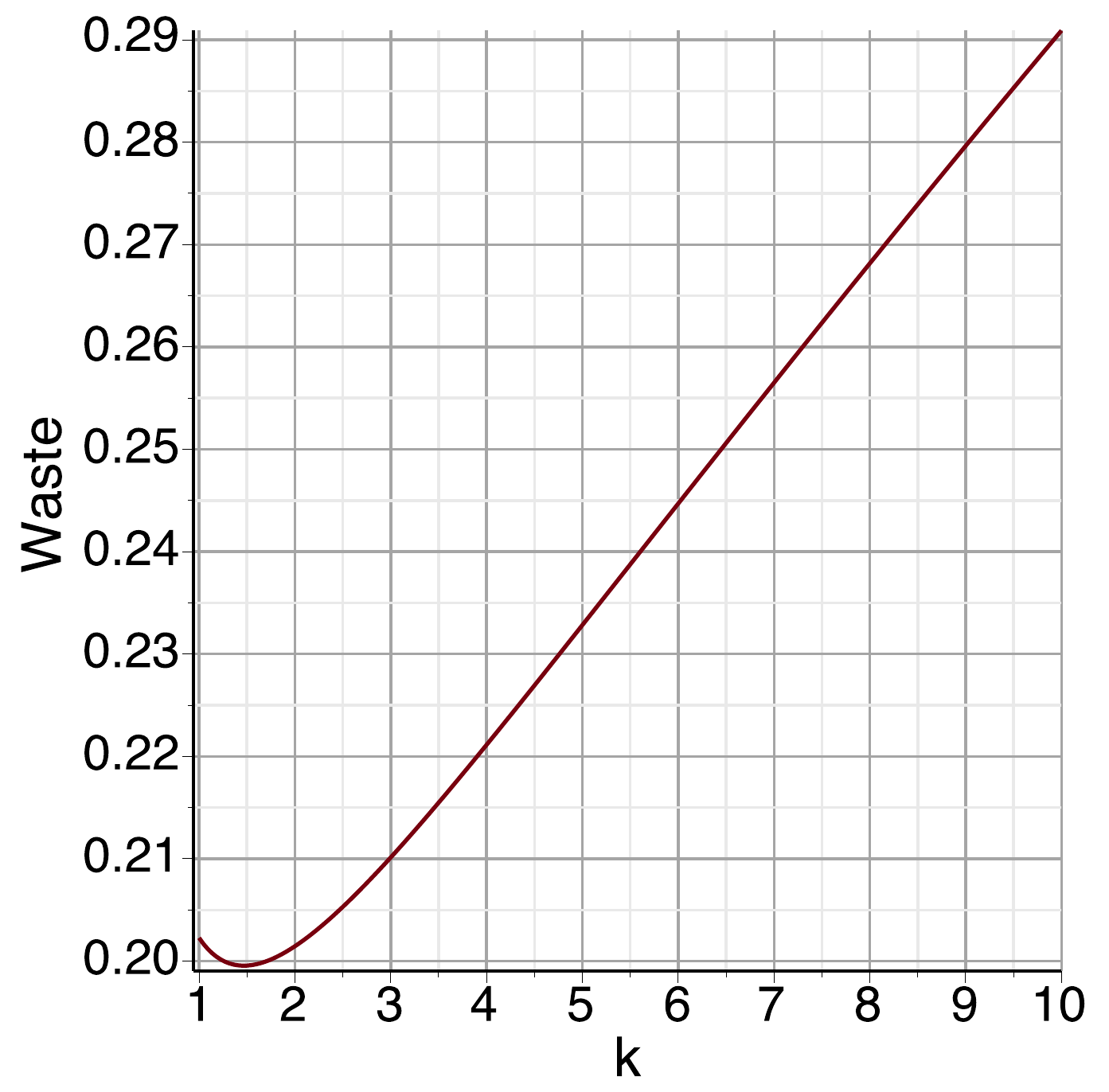}
\includegraphics[width=.45\linewidth]{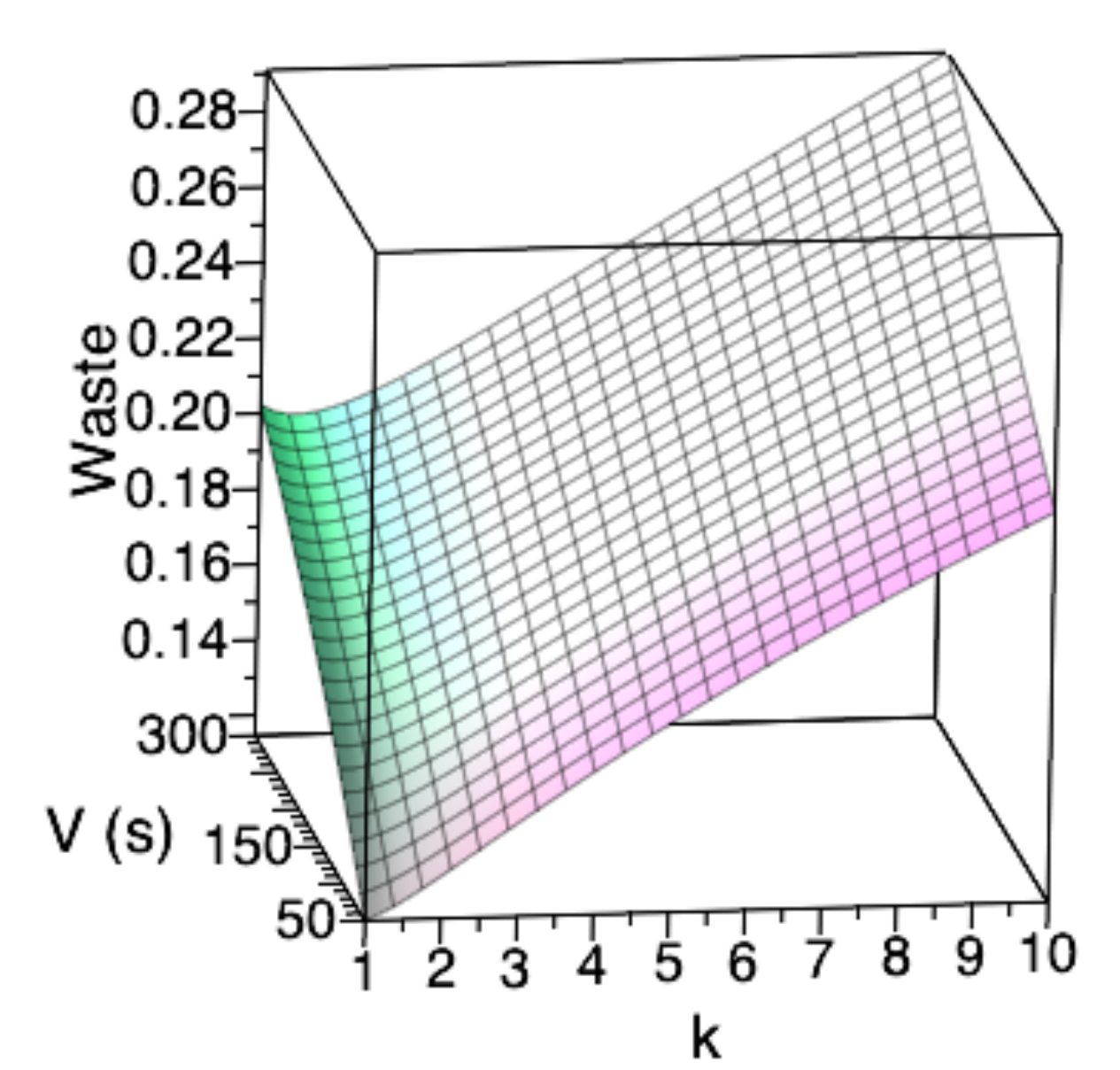}
\caption{Case with $k$ checkpoints, and one verification per periodic
  pattern. Waste as function of $k$, and potentially of $V$, using the
  optimal period.  \footnotesize{($\vvv=300s, \ccc=\rrr=60s, \ddd=0s$,
    and $\muplatform = \frac{10y}{10^5}$.)}
\label{fig:maple:kC1V2}\vspace{-2\baselineskip}}
\end{center}
\end{figure}

This is illustrated even more clearly with Figure~\ref{fig:maple:kC1V2}, 
where the checkpoint costs and machine
availability are set to the second scenario of 
Sections~\ref{sec.evaluation.k} and~\ref{sec.evaluation.kv1c}. 
As soon as the checkpoint cost is not negligible compared to
the verification cost (only 5 times smaller in this case), it is more
efficient to validate every other checkpoint than to validate only after
$k > 2$ checkpoints. The 3D surface shows
that this holds true for rather large values of \vvv.

All the rollback / recovery techniques that we
have evaluated above, using various parameters for the different
costs, and stressing the different approaches to their limits,
expose a waste that remains, in the vast majority of the cases,
largely below $66\%$. This is noticeable, because the traditional
hardware based technique, which relies on triple modular redundancy
and voting~\cite{Lyons1962}, mechanically presents  a waste that is at least
equal to $66\%$ (two-thirds of resources are wasted, and neglecting the cost of voting).

\section{Related work}
\label{sec.related}

As already mentioned, this work is motivated by the recent paper by 
Lu, Zheng and Chien~\cite{LuZhengChien2013}, who introduce a \emph{multiple checkpointing model} 
to compute the optimal checkpointing period with error detection latency. 
We start with a brief overview of traditional checkpointing approaches
before discussing error detection and recovery mechanisms. 

\subsection{Checkpointing}

Traditional (coordinated) checkpointing has been studied for many
years.  The major appeal of the coordinated approach is its
simplicity, because a parallel job using $n$ processors of individual
MTBF $M_{ind}$ can be viewed as a single processor job with MTBF
$\muplatform = \frac{M_{ind}}{n}$. Given the value of $\muplatform$, an approximation of
the optimal checkpointing period can be computed as a function of the
key parameters (downtime $\ddd$, checkpoint time $\ccc$,
and recovery time $\rrr$). The first estimate had been given by
Young~\cite{young74} and later refined by Daly~\cite{daly04}. Both
use a first-order approximation for Exponential failure
distributions; their derivation is similar to the approach in
Equations~\eqref{eq.tfinal} and~\eqref{eq.waste}.  More accurate
formulas for Weibull failure distributions are provided in
~\cite{ling2001variational,ozaki2006distribution,doi:10.1007/978-3-642-14390-8_22}.
The optimal checkpointing period is known only for Exponential failure
distributions~\cite{c178}.  Dynamic programming heuristics for
arbitrary distributions are proposed
in~\cite{toueg1983optimum,10.1109/TC.2012.57,c178}.

The literature proposes different
works~\cite{Plank01processorallocation,SunICPP10,Wang_DSN05,4367962,5289177}
on the modeling of coordinated checkpointing
protocols. In particular,  \cite{SunICPP10} and~\cite{Plank01processorallocation}
focus on the usage of available resources: some may be kept as backup
in order to replace the down ones, and others may be even shutdown in
order to decrease the failure risk or to prevent storage consumption
by saving fewer checkpoint snapshots.

The major drawback of coordinated checkpointing protocols is their
lack of scalability at extreme-scale.  These protocols will lead to
I/O congestion when too many processes are checkpointing at the same
time. Even worse, transferring the whole memory footprint of an HPC
application onto stable storage may well take so much time that a
failure is likely to take place during the transfer!  A few
papers~\cite{5289177,j116} propose a scalability study to assess the
impact of a small MTBF (i.e., of a large number of
processors). The mere conclusion is that checkpoint time should be
dramatically reduced for platform waste to become
acceptable, which motivated the instantiation of optimistic scenarios in
Section~\ref{sec.evaluation}.

All the above approaches maintain a single checkpoint. 
If the checkpoint file includes errors, the application 
faces an irrecoverable failure and must restart from scratch. This is because error 
detection latency is ignored in traditional rollback and recovery schemes. These schemes
assume instantaneous error detection (therefore mainly targeting fail-stop failures)
and are unable to accommodate silent errors.

\subsection{Error detection}

Considerable efforts have been directed at error-checking to reveal latent errors.
Most techniques combine redundancy at various levels, together with a variety
of verification mechanisms. The oldest and most drastic approach is at the hardware level,
where all computations are executed in triplicate, and majority voting is enforced in case
of different results~\cite{Lyons1962}. Error detection approaches include 
memory scrubbing~\cite{Hwang2012}, fault-tolerant 
algorithms~\cite{Bronevetsky2008,Heroux2011,Shantharam2011}, ABFT 
techniques~\cite{Kuang1984,Bosilca2009} and critical MPI message validation~\cite{Fiala2012}.
We refer to Lu, Zheng and Chien~\cite{LuZhengChien2013} for a comprehensive list
of techniques and references. As already mentioned, our work is agnostic of the underlying
error-detection technique and takes the cost of verification as an input parameter to the model
(see Section~\ref{sec.ourmodel}).

\section{Conclusion}
\label{sec.conclusion}

In this paper, we revisit traditional checkpointing  and rollback recovery strategies.
Rather than considering fail-stop failures, we focus on silent data corruption errors.
Such latent errors cannot be neglected anymore in High Performance Computing, in particular
in sensitive and high precision simulations. The core difference with fail-stop failures is
that error detection is not immediate.

We discuss and analyze two models. In the first model, errors are detected after some delay following a
probability distribution (typically, an Exponential distribution). We compute the optimal 
checkpointing period in order to minimize the waste when all checkpoints can be kept in memory,
and we show that this period does not depend on the distribution of detection times. In practice,
only a few checkpoints can be kept in memory, and hence it may happen
that an error was detected after
the last correct checkpoint was removed from storage. We derive a minimum value of the period
 to guarantee, within a risk threshold,  that at least one valid
 checkpoint remains when a latent error is detected. 

A more realistic model assumes that errors are detected through some verification mechanism.
Periodically, one checks whether the current status is meaningful or not, and then eventually detects
a latent error. We discuss both the case where the periodic pattern includes $k$ checkpoints 
for one verification (large cost of verification), and the opposite case with $k$ verifications
for one checkpoint (inexpensive cost for verification). We express a formula for the waste in both cases,
and, from these formulas, we derive the optimal period. 

The various models are instantiated with realistic parameters, and the evaluation results 
clearly corroborate the theoretical analysis. For the first model, with detection times, 
the tradeoff between waste and risk of irrecoverable error clearly appears, hence showing that
a period larger than the one minimizing the irrecoverable-failure-free waste should often be chosen 
to achieve an acceptable risk. 
The advantage of the second model is that there are no irrecoverable failures (within each period,
there is a verification followed by a checkpoint, hence ensuring a valid checkpoint). 
We compute the optimal pattern of checkpoints and verifications per period, as a function of their
respective cost, to minimize the waste. The pattern with more checkpoints than verification turns out
to be usable only when the cost of checkpoint is much lower than the cost of verification, 
and the conclusion is that it is often more efficient to verify the result every other checkpoint. 

Overall, we provide a thorough analysis of checkpointing models for latent errors,
both analyzing the models analytically, and evaluating them through different scenarios. 
A future research direction would be to study more general scenarios of multiple checkpointing,
for instance by keeping not the consecutive $k$ last checkpoints in the first model, but rather 
some older checkpoints to decrease the risk. In the second model, more 
verification/checkpoint combinations could be studied, while we focused on cases 
with an integer number of checkpoints per verification (or the converse). 

\section*{Acknowledgments}
This work was supported in part by the ANR {\em RESCUE} project. 
A.~Benoit and Y.~Robert are with the Institut Universitaire de France.

\bibliographystyle{IEEEtran}
\bibliography{biblio}
\end{document}

%% file: macro-fig.tex
\usepackage{tikz-timing}

\usetikztiminglibrary[new={char=Q,reset char=R}]{counters}
\usetikzlibrary{patterns,arrows,decorations.pathreplacing}

\newcommand{\faultbis}[1]{
\draw[<-, color=red] (#1) -- ($(#1)+(0.2,1.2)$) -- ($(#1)+(0.1,1.4)$) --  ($(#1)+(0.2,2)$) node[above, left] {\scriptsize{Error}};
} 
\newcommand{\detecfault}[1]{
\draw[<-, color=blue] (#1) -- ($(#1)+(0.2,1.2)$) -- ($(#1)+(0.1,1.4)$) --  ($(#1)+(0.2,2)$)  node[above, right] {\scriptsize{Detection}};
} 

\newcommand{\legende}[3]{
\draw[thick, <->] ($(#1)+(0,-0.20)$) -- ($(#1)+(#2,-0.20)$) node[below=-0.5pt, midway] {\scriptsize{#3}};
}
\newcommand{\semilegende}[3]{
\draw[thick,dashed,<-] ($(#1)+(0,-0.20)$) -- ($(#1)+(1,-0.20)$) node[below=-0.5pt, midway] {};
\draw[thick,->] ($(#1)+(1,-0.20)$) -- ($(#1)+(#2,-0.20)$) node[below=-0.5pt, midway] {\scriptsize{#3}};
}

\newcommand{\arrowtime}[2]{
\draw[thick, color=black,->] (0,#1) -- (#2,#1) node[below=-0.5pt, ] {\scriptsize{Time}};
}

\newcommand{\ttrd}{4} 

\newcommand{\patternCV}[3]{
\draw[very thick, color=red] ($(#1,#3)+(0,-0.8)$) -- ($(#1,#3)+(0,1.8)$);
\draw[very thick, color=red] ($(#2,#3)+(0,-0.8)$) -- ($(#2,#3)+(0,1.8)$);
}

%% file: fig/fail-detect-latency.tex
\centering
\begin{tikztimingtable}[
    timing/slope=0,         
    timing/rowdist=\ttrd,     
    timing/coldist=2pt,     
    xscale=2,yscale=1.5, 
    semithick ,              
  ]

&\\
\extracode
 \makeatletter
 \begin{pgfonlayer}{background}
\arrowtime{0}{12};
\semilegende{0,0}{3.6}{$X_e$};
\legende{3.6,0}{3}{$X_d$};
\faultbis{3.6,0};
\detecfault{6.6,0};

 \end{pgfonlayer}
\end{tikztimingtable}%



%% file: fig/pattern5C1V.tex
\centering
\begin{tikztimingtable}[
    timing/slope=0,         
    timing/rowdist=\ttrd,     
    timing/coldist=2pt,     
    xscale=3,yscale=1.5, 
    semithick ,              
  ]

& 0.2S{G}0.8D{\vvv}0.5D{\ccc}{[red] G1.3L0.5D{\ccc}1.3L0.5D{\ccc}1.3L 0.5D{\ccc}1.3L0.5D{\ccc}1.3L0.8D{\vvv}0.5D{\ccc}G} \\
\extracode
 \makeatletter
 \begin{pgfonlayer}{background}
\arrowtime{-0}{12};
\patternCV{1.5}{11.3}{0};
\legende{1.5,-0}{1.3}{\www};
\legende{3.3,-0}{1.3}{\www};
\legende{5.1,-0}{1.3}{\www};
\legende{6.9,-0}{1.3}{\www};
\legende{8.7,-0}{1.3}{\www};
 \end{pgfonlayer}
\end{tikztimingtable}%



%% file: fig/pattern1C5V.tex
\centering
\begin{tikztimingtable}[
    timing/slope=0,         
    timing/rowdist=\ttrd,     
    timing/coldist=2pt,     
    xscale=3,yscale=1.5, 
    semithick ,              
  ]
  
& 0.2S{G}0.3D{\vvv}0.5D{\ccc}{[red] G1.7L0.3D{\vvv}1.7L0.3D{\vvv}1.7L 0.3D{\vvv}1.7L0.3D{\vvv}1.7L0.3D{\vvv}0.5D{\ccc}G} \\
\extracode
 \makeatletter
 \begin{pgfonlayer}{background}


\arrowtime{0}{12};
\patternCV{1}{11.5}{0};
\legende{1,0}{1.7}{\www};
\legende{3,0}{1.7}{\www};
\legende{5,0}{1.7}{\www};
\legende{7,0}{1.7}{\www};
\legende{9,0}{1.7}{\www};
 \end{pgfonlayer}
\end{tikztimingtable}%

